\documentclass[format=acmsmall,review=false,screen=true]{acmart}

\startPage{1}
\fancypagestyle{firstpagestyle}{
	\fancyhf{}%
	\fancyfoot[L]{}%
}
\fancypagestyle{acmec}{
	\fancyhf{}%
	\fancyhead[C]{\footnotesize{\mbox{\shortauthors}}}
	\fancyhead[R]{\thepage}%
}
\renewcommand{\footnotetextcopyrightpermission}[1]{\thankses}

\usepackage{booktabs} 
\usepackage{amsfonts,amsmath,amsthm,amssymb,amscd}
\usepackage{array,paralist,tikz,setspace,xcolor}
\usetikzlibrary{patterns,positioning,arrows.meta,calc}
\usepackage{epsfig,url}
\usepackage{color,colortbl}
\usepackage{dsfont}

\usepackage{wrapfig}
\usepackage{multirow}
\usepackage{hyperref}

\usepackage{makecell}
\usepackage{enumitem}
\newcommand{\PreserveBackslash}[1]{\let\temp=\\#1\let\\=\temp}
\newcolumntype{C}[1]{>{\PreserveBackslash\centering}p{#1}}

\hypersetup{citecolor=red,urlcolor=black!70!blue,linkcolor=blue!70!red,colorlinks=true}

\def\T{\mathcal{T}} 
\def\Tn#1{\mathcal{T}^{(#1)}} 
\def\Tstar{\mathcal{T}^{(*)}} 
\def\P{N} 
\def\PoT{\mathrm{PoT}} 
\def\PoA{\mathrm{PoA}} 
\def\vs{\textbf{s}}

\def\G{\mathcal{G}}
\def\vec#1{\textbf{#1}}
\renewcommand{\epsilon}{\varepsilon}

\newcommand{\s}{\vec{s}}
\newcommand{\R}{\mathds{R}}

\newcommand{\ycomment}[1]{}
\setlist[description]{font=\normalfont\itshape,itemsep=0mm}

\begin{document}

\title{The Impact of Tribalism on Social Welfare}

\author{Seunghee Han}
\affiliation{%
  \institution{Department of Computer Science, Cornell University}
  \country{USA}
}
\email{sh969@cornell.edu}
\author{Matvey Soloviev}
\affiliation{%
  \institution{Department of Computer Science, Cornell University}
  \streetaddress{107 Hoy Rd}
  \city{Ithaca}
  \state{NY}
  \postcode{14853}
  \country{USA}}
\email{msoloviev@cs.cornell.edu}
\author{Yuwen Wang}
\affiliation{%
  \institution{Department of Mathematics, Cornell University}
  \streetaddress{310 Malott Hall}
  \city{Ithaca}
  \state{NY}
  \postcode{14853}
  \country{USA}}
\email{ywang@math.cornell.edu}

\begin{abstract}
We explore the impact of mutual altruism among the players belonging to the same set -- their \emph{tribe} -- in a partition of all players in arbitrary strategic games upon the quality of equilibria attained. To this end, we introduce the notion of a \emph{$\tau$-tribal extension} of an arbitrary strategic game, in which players' subjective cost functions are updated to reflect this, and the associated \emph{Price of Tribalism}, which is the ratio of the social welfare of the worst Nash equilibrium of the tribal extension to that of the optimum of social welfare. We show that in a well-known game of friendship cliques, network contribution games as well as atomic linear congestion games, the Price of Tribalism is higher than the Price of Anarchy of either the purely selfish players or fully altruistic players (i.e. ones who seek to maximise the social welfare). This phenomenon is observed under a variety of equilibrium concepts.
In each instance, we present upper bounds on the Price of Tribalism that match the lower bounds established by our example.

\end{abstract}

\maketitle

\section{Introduction}
\def\E{\mathbb{E}}

According to the standard narrative around the concept of Nash equilibrium, one of its great contributions is that it shed light on why multi-agent systems in the real world often ``race to the bottom'', or otherwise fail to exhibit behaviour anywhere near the social optimum.
Perhaps fortunately, the picture suggested by the theory is not always reflected in real-world systems, which often appear to stabilise in states that are better than self-interested Nash equilibria. On the other hand, many a carefully designed political and economic system fails to deliver on its theoretical promises in reality.

Beside the ``spherical cow'' class of model-reality disagreements such as computational limitations and insufficient rationality of the agents,
the good half of this discrepancy is often rationalised by saying that players exhibit a degree of altruism 
-- that is, they seek to optimise not just their own welfare but some weighted combination of it with the sum of the welfare of all players. This approach may partially explain the more favourable dynamics we observe. However, recent results \cite{cara,chen} demonstrate that in some cases, altruism can give rise to equilibria that are even worse than those that exist if all players are purely self-interested.

Once identified, this might not seem that unrealistic: for instance, all historical industrial revolutions, implemented by arguably selfish agents seeking to maximise profit, were accompanied by a temporary dip in social welfare \cite{industrialisation}. A society whose members are altruistic but do not coordinate may therefore never have implemented these changes, remaining stuck in the pre-industrial local optimum without electricity, mass production or modern medicine.

Looking at the workings of the lower bounds for ``altruistic anarchy'', we find that the ways in which altruistic players get stuck in local optima appear quite different from those enabling bad selfish equilibria. Could it then be that there are realistic settings in which both mechanisms occur together?

Tribalism and political polarisation are often argued to be a feature of human interactions that predates those interactions even being, strictly speaking, human at all, but according to news media and sociological analyses alike, their impact on public life in Western societies has been steadily increasing
\cite{economist,polarisation-twitter,american-polarisation,polarisation-dynamics}. A tribalistic agent, broadly speaking, is concerned with the welfare of other agents belonging to the same tribe, rather than the overall social welfare of everyone participating in the system. A game or system in which the agents are tribalistic, then, could be said to exhibit both altruism (within a tribe) and selfishness (in how the tribes interact with each other). Inspired by the failure of real-world systems, we set out to investigate if this mixture of altruism and selfishness could in fact lead to even worse outcomes than either altruism or selfishness alone.

We find that this is indeed the case. In the following sections, we will show that tribalism leads to a greater Price of Anarchy than either altruism or selfishness in 
a folklore model of friend cliques, network contribution games \cite{contribution-games} and atomic linear routing games \cite{og-atomicrouting}. In each case, we also present upper bounds for the tribal Price of Anarchy that match the lower bounds demonstrated by our examples.

\section{Main results}

\subsection{Definition of tribalism}

Our definition of games in which the players exhibit tribalism is designed to resemble the definition of $\alpha$-\emph{altruistic extensions} that
 \cite{chen} make for their analysis of universal altruism.
An analogous definition can be made for utility-maximisation games.

We will represent cost-minimisation games $G$ as the triple $(\P, (\Sigma_i)_{i \in \P}, (c_i)_{i\in \P})$, where $\P$ is the
set of players, $\Sigma_i$ are the strategies available to player $i$ and $c_i(\vec{s})$ is the cost for player $i$
when the vector of strategies chosen by all players is $\vec{s}\in \prod_{i\in \P} \Sigma_i$. We will use $(t;\vec{s}_{-i})$
to denote the vector $\vec{s}$ with player $i$'s entry replaced with $t$.

\begin{definition} Suppose $G=(\P, (\Sigma_i)_{i \in \P}, (c_i)_{i\in \P})$ is a finite cost-minimisation game.

Let $\tau:\P\rightarrow \mathds{N}$
be a function that assigns each player a unique tribe, identified by a natural number. The \emph{$\tau$-tribal extension of $G$}
is the cost-minimisation game $G^\tau = (\P, (\Sigma_i)_{i\in \P}, (c_i^\tau)_{i\in \P})$, where the cost experienced
by every player is the sum of costs of all players in the same tribe in the original game:
for every $i\in \P$ and $\vec{s}\in \Sigma = \Sigma_1 \times \dotsb \times \Sigma_n$,
$$ c_i^\tau(\vec{s}) = \sum_{j\in \P: \tau(i)=\tau(j)} c_j(\vec{s}). $$
\label{defn:tribalexten}
\end{definition}

 When the partition function $\tau$ is constant, our definition agrees with the one in \cite{chen} with $\alpha=1$, and we say the players are \emph{(fully) altruistic}. When $\tau(i)\neq \tau(j)$ for all $i\neq j$, $G^\tau=G$ and we say the players are \emph{selfish}.

We then define the Price of Tribalism for a class of games $\G$ and class of partition functions $\{\T_G\}_{G\in \G}$  as the supremum of ratios between the social welfare $C_G(\vec{s})=\sum_{i\in \P} c_i(\vec{s})$ of any Nash equilibrium \emph{of any $\tau$-tribal extension} for $\tau\in \T_G$
and the social optimum, i.e. the lowest attainable social cost. In other words, the PoT captures how bad a ``tribal equilibrium'' can get for any game in $\G$ and any pattern of tribal allegiance of the players therein.

The definition for utility-maximisation games is again analogous.

\begin{definition}
The pure (resp. correlated, strong, mixed...) tribal Price of Anarchy (\emph{Price of Tribalism}, PoT) of a class of games $\mathcal{G}$ and class of partition functions for each game $\T=\{\T_G\}_{G\in \mathcal{G}}$ to be
$$ \PoT(\T,\mathcal{G}) = \sup_{G\in \mathcal{G}, \tau\in \T_G} \frac{ \sup_{\vec{s}\in S_{G^\tau}} C_G(\vec{s}) }{ \inf_{\vec{s}\in \Sigma} C_G(\vec{s}) }, $$
where $S_{G^\tau}$ is the set of pure (correlated, strong, mixed...) Nash equilibria of $G^\tau$.
\label{defn:pot}
\end{definition}

We can control the tribal structures that we want to consider by choosing an appropriate class $\T$ of partition functions. We will denote the class of all functions which sort the players into exactly $k$ tribes by

$\Tn{k}_G = \{ \tau : |\tau(\P)|=k \}$, and that of all possible functions as $\Tstar_G = \bigcup_{i=1}^\infty \Tn{k}_G$. The Price of Anarchy given full altruism (i.e. in the game $G^1$ of \cite{chen}) then equals $\PoT(\Tn{1},\mathcal{G})$.

\subsection{Impact of tribalism on known games}

We first consider a folklore game 
 which is often invoked as a simple model of friendship.
This game can be seen as a special case of the \emph{party affiliation game} of \cite{fabrikant}, with positive payoffs only.
Players choose to associate with one of two cliques $A$ and $B$. Each pair of players has an associated utility of being friends $u_{ij}$, which they can only enjoy if they choose to associate with the same clique. For this game, a folklore bound we 
revisit shows that the pure Price of Anarchy is $2$.

\begin{theorem} (\ref{thm:socialaltruism}; \ref{thm:pot2grouping}; \ref{thm:potkgrouping}) The pure Price of Anarchy for the social $2$-grouping game $\mathcal{F}_2$ under full altruism satisfies $$\PoT(\Tn{1},\mathcal{F}_2)=2$$ as well. However, the pure Price of Tribalism is $$\PoT(\Tstar,\mathcal{F}_2)=3.$$ 
    More generally, in the social $k$-grouping game (i.e. with $k$ cliques) with at least $k$ tribes $\mathcal{F}_k$,  
$$\PoT(\Tn{1},\mathcal{F}_k)=k \text{\hspace{1em}and\hspace{1em}} \PoT(\Tstar,\mathcal{F}_k)=2k-1,$$
while the pure Price of Anarchy is $k$.
  \label{thm:friendbounds}
\end{theorem}

\begin{table*}
\renewcommand{\c}[1]{ {\color{gray} #1} }
\renewcommand{\arraystretch}{1.2}
\setlength{\tabcolsep}{10pt}
\small
\begin{tabular*}{\textwidth}{C{3.2cm} | l l l }
  \toprule
  Game   &PoA     &Altruistic PoA  & $\PoT$ \\
  \midrule
  Social grouping game with 2 cliques & 2 \c{(folklore)}
                                      & 2 \c{(Thm. \ref{thm:socialaltruism})}
                                      & 3 \c{(Thm. \ref{thm:pot2grouping})} \\
  Social grouping game with $k$ cliques & $k$ \c{(Thm. \ref{thm:friendshipk})}
                                        & $k$ \c{(Thm. \ref{thm:friendshipk})}
                                        & \makecell[l]{$2k -1$ \c{(Thm. \ref{thm:potkgrouping})} } \\
  Network contribution with additive rewards & 1 \c{\cite{contribution-games}}
                                                  & 1 \c{(\cite{contribution-games}, Cor. \ref{cor:NSGsum1})}
                                                  & 2 \c{(Thm. \ref{lem:NCGsum2})} \\
  Network contribution with convex rewards & 2 \c{\cite{contribution-games}}
                                                            & 2 \c{(Thm. \ref{thm:NCGaltconvex})}
                                                            & 4 \c{(Thm. \ref{lem:NCGconv4})} \\
Atomic linear routing & $5/2$ \c{\cite{ChristodoulouK05,a-atomic-routing}}
                      & 3 \c{\cite{cara}}
                      & \makecell[l]{4 \c{(Thm. \ref{thm:alrg3k})}}  \\
  \bottomrule
\end{tabular*}

\label{tab:summary}
\end{table*}
A more involved model of friendship networks was first described by Anshelevich and Hoefer \cite{contribution-games}. In the \emph{network contribution game} $\mathcal{N}_\mathcal{F}$, we are given a social graph of vertices representing players and edges representing potential relationships between them. Each player has a fixed budget $b_i$, which they seek to allocate among the edges adjacent to them. They then receive a payoff from each edge based on a symmetric function $f_e(x,y)=f_e(y,x)\in \mathcal{F}$ of their own and the other player's contribution to that edge.

In the original paper, the authors show different bounds on the Price of Anarchy for this game depending on the form the functions $f_e$ can take: among others, for $f_e(x,y)=c_e(x+y)$ (we call this class of games $\mathcal N_+$), they show a PoA of $1$, and when $f_e(x,y)$ satisfies $f_e(x,0)=0$ and each $f_e$ is convex in each coordinate (denoted by $\mathcal N_C$), the PoA is $2$. 
Moreover, instead of pure Nash equilibria, they invoke \emph{pairwise} ones, which are resilient against any pair of associated players deviating together.
In the presence of tribalism, we demonstrate that both of these bounds deteriorate.

\begin{theorem} (\ref{lem:NCGsum2}; \ref{lem:NCGconv4}; \ref{thm:NCGaltconvex}) The pure and pairwise Price of Tribalism for the network contribution game with additive rewards is $$\PoT(\Tstar,\mathcal{N}_+)=2.$$

The pure and pairwise Price of Tribalism for each of the network contribution games with
coordinate-convex reward functions is
$$ \PoT(\Tstar,\mathcal{N}_{C})=4. $$
Meanwhile, the altruistic Price of Anarchy is still $1$ for $\mathcal{N}_+$ and $2$ for $\mathcal N_C$.
\end{theorem}

Finally, we will turn our attention to atomic linear routing games $\mathcal{R}$ \cite{rose-1973}, a popular class of games that model a set of players seeking to each establish a point-to-point connection over a shared network (such as the internet or roads) represented by a graph, where the cost to all players using an edge increases linearly with the number of players using it. In the case of selfish behaviour, these games are well-known to exhibit a pure Price of Anarchy of $5/2$ \cite{rose-1973}. Caragiannis et al. \cite{cara} show that in the case of universal altruism, this deteriorates to $\PoT(\Tn{1},\mathcal{R})=3$. We show matching lower and upper bounds that demonstrate that in the face of tribalism, significantly worse equilibria can arise. It should be noted that in fact, our result applies to the more general class of atomic linear congestion games.

\begin{theorem} (\ref{thm:alrg3k}) The Price of Tribalism for the atomic linear routing game is
$$\PoT(\Tn{2},\mathcal{R})=\PoT(\Tstar,\mathcal{R})=4.$$
\label{thm:routingbounds}
\end{theorem}

\section{Background}

The Price of Anarchy \cite{anarchy} is a widely used tool for analysing the behaviour of systems composed of autonomous agents pursuing their rational self-interest in the absence of a central coordination mechanism. However, the assumption of pure selfishness is at odds with observed human behaviour and what is predicted from evolutionary biology; social scientists have conducted experiments and offered simple models for altruism \cite{publicgood,modelaltruism}.

This observation inspires a well-explored line of inquiry

\cite{cara,chen2,chen} into how players caring about social welfare affects the Price of Anarchy. In \cite{cara}, the players of an atomic linear congestion game are taken to behave altruistically on a spectrum from zero to one: 0 being purely selfish, $1/2$ corresponding players trying to maximise social welfare and $1$ making players ``totally selfless'', so that they optimise for everybody's utility except their own.
Surprisingly, it is found that altruism actually often leads to \emph{higher} prices of anarchy than if the players behaved selfishly.

We draw heavy inspiration from the subsequent treatment by \cite{chen}, which applies a similar construction to a broader set of games. 

The authors of this work
introduce the notion of an $\alpha$-\emph{altruistic extension} of a given game, where each player $i$ has an associated altruism parameter $\alpha_i$. The $i$th player's cost is defined as a combination of $(1-\alpha_i)$ times their individual utility and $\alpha_i$ times the social welfare, so a value of $0$ represents selfishness while $1$ corresponds to full altruism. (Unlike \cite{cara}, \cite{chen} do not consider selfless players.) When $\alpha_i=1$ for all $i$, this corresponds to our notion of tribalism with every player belonging to the same tribe; however, we do not consider the possibility of partial altruism in the sense of a parameter $\alpha$ strictly between $0$ and $1$.

Many authors \cite{collusion,strongPoA,dimitri_coalition,oliroute,civilsoc} have studied how various forms of ``group behaviour'' can affect the Price of Anarchy. Of these, \cite{collusion,civilsoc} consider the ratio between the worst Nash equilibrium of the modified game to the worst selfish equilibrium (contrast to our ratio between the worst modified Nash and the social optimum), and only look at specific games (congestion games in \cite{collusion} and load balancing games in \cite{civilsoc}). The \emph{Price of Collusion} defined in \cite{collusion} in particular is closely related to the instantiation of our definition for those games, with $\mathrm{PoC} \leq \PoT \leq \mathrm{PoA}\cdot \mathrm{PoC}$.
In \cite{strongPoA,dimitri_coalition,oliroute}, the groups are taken to be able to coordinate their actions (unlike ours), but (with the exception of \cite{oliroute}) not share costs (so players are still selfish).

Much of the aforementioned work implicitly or explicitly assumes altruism to be a phenomenon that mostly occurs locally in tightly-knit groups, where it is reasonable to assume that players who care for each other's welfare can also coordinate their actions and deviate from a strategy together.
We generally think of tribes as entities in which, due to scale or other impediments, coordination between individuals is not possible (contrast with e.g. \cite{PoD}, where the groups vote on their next action), though some of our results on \emph{coordinated} Price of Tribalism in network contribution games show that bad equilibria can persist even when individuals in a tribe coordinate.
We contend that this is a reasonable assumption: human political and religious groups rarely possess any sort of central coordination mechanism (small cells or action groups that do act in concert could easily be merged into a single ``superplayer'' to be considered in our setting), and there are natural examples such as state-forming insects like bees and ants where non-coordination is implied by the absence of a communication system rich enough to express the breadth of actions available to each individual. (Cues that \emph{are} known to be used for some degree of coordination, such as scent markers, can be modelled as part of the game's payoff structure.)

A different strain of work that is based on the idea that altruism is a local phenomenon dates back to \cite{graphmodels}. This paper considers games in which the players' perceived utilities are only affected by their
own and those of their neighbours in the \emph{social graph}, representing a geographic or social proximity that leads individuals to be concerned with each other's utility.
 Related lines of work extend the social graph to model various other relations between players such as information about each other's actions and preferences (e.g. \cite{graphical-congestion,ignorance}).
The social graph may be taken to be closely related to some graph structure present in the underlying game, as in the disease spread model of \cite{windfall}, or on the contrary to be completely independent, as in the \emph{social context games} introduced in \cite{itai_social_context}. The latter, in particular, turn out to be a strictly more general model than the one considered in this paper: each player is made to minimise a fixed function of the costs experienced
by themselves and all their immediate neighbours in the social graph (called \emph{social context graph} by \cite{itai_social_context}).

A $\tau$-tribal extension can be seen as a social context game where the social graph 
consists of a disjoint collection of cliques, i.e. complete subgraphs that are not connected to each other, and the aggregator function is a sum.

Although the model is on the surface quite general,

it has been found to exhibit a good amount of exploitable structure:
for example, \cite{SCsmoothness} defines smoothness for social context games, which generalises the original definition of \cite{smooth-roughgarden}.

\cite{SCpotential} tightly characterises which potential games remain potential games under all social contexts. One such example are atomic linear routing games, which we consider in Section \ref{sect:atomiclinear}.

Several authors \cite{bilo,SClinearcong,rewardsharing} explore how the Price of Anarchy behaves in this more general setting, which enables still more mechanisms to construct games with bad equilibria -- for instance, player's concern need no longer be transitive, so player A can seek to reduce the cost of player B and player B that of player C without player A having any concern for player C's cost.
Relevantly, \cite{bilo} demonstrates a lower bound and upper bound of $17/3$ for atomic linear congestion games with social context given by a general graph, which in our setting (i.e. when the graph is a collection of disjoint cliques) have a matching lower and upper bound of 4.
Due to the generality of the social context model, we consider our model to be of independent interest.

\section{Examples}

\subsection{Social grouping games}

A simple folklore example of a game with nontrivial and clearly suboptimal Nash equilibria is the \emph{social ($k$-)grouping game}. In this game, the players are nodes of a directed graph, with edge weights $u_{ij}\geq 0$ to be thought of as the benefit a friendship between players $i$ and $j$ would give player $j$. We can assume the graph to be complete, with previously absent edges having weight 0. In general, we do not assume $u_{ij}=u_{ji}$, but both our lower and upper bounds satisfy this assumption. Each player must declare their membership in one of two, or more generally $k$ ``cliques'' or ``friend groups'', and receives utility $u_i(\s)=\sum_{j:\s(i)=\s(j)} u_{ji},$ that is, the sum of benefits from all other players in the same clique.

Here, we will focus on the lower bounds $2$-clique case; for the fully general version of these theorems and proofs of the upper bounds, see Appendix \ref{sec:a.social}.

It is not hard to see that in this game, it is optimal for everyone to declare membership in the same clique, therefore being able to reap the benefits of all possible friendships. However, there exist locally optimal pure Nash equilibria that fall short of the optimum by up to a factor of two. The following theorem is known to us from private communication with {\'{E}}va Tardos. 

\begin{theorem} \label{thm:socialhomework}
 The pure Price of Anarchy for the social $2$-grouping game is $2$.
\end{theorem}

\begin{figure}
\begin{center}

\begin{tikzpicture}[scale=1.1,every node/.append style={font=\scriptsize}]
 \node [inner sep=0.1em]  (a) at (0,0) {\normalsize a};
 \node [inner sep=0.1em]  (b) at (2,0) {\normalsize b};
 \node [inner sep=0.1em]  (c) at (2,1) {\normalsize c};
 \node [inner sep=0.1em]  (d) at (0,1) {\normalsize d};

 \node (A) at (-0.8,0.5) {\normalsize A};
 \draw (0,0.5) ellipse (0.5 and 1);

 \node (B) at (2.8,0.5) {\normalsize B};
 \draw (2,0.5) ellipse (0.5 and 1);

 \draw[->] (a) to [bend left=20] node[above] {$1$} (b);
 \draw[->] (b) to [bend left=20] node[left] {$1$} (c);
 \draw[->] (c) to [bend left=20] node[below] {$1$} (d);
 \draw[->] (d) to [bend left=20] node[right] {$1$} (a);
 \draw[<-] (a) to [bend right=20] node[below] {$1$} (b);
 \draw[<-] (b) to [bend right=20] node[right] {$1$} (c);
 \draw[<-] (c) to [bend right=20] node[above] {$1$} (d);
 \draw[<-] (d) to [bend right=20] node[left] {$1$} (a);

\end{tikzpicture}
\hspace{1em}
\begin{tikzpicture}[scale=1.1,every node/.append style={font=\scriptsize}]
 \node [inner sep=0.1em]  (a) at (0,0) {\normalsize \color{red} a};
 \node [inner sep=0.1em]  (b) at (2,0) {\normalsize \color{blue} b};
 \node [inner sep=0.1em]  (c) at (2,1) {\normalsize \color{blue} c};
 \node [inner sep=0.1em]  (d) at (0,1) {\normalsize \color{red} d};

 \node (A) at (-0.8,0.5) {\normalsize A};
 \draw (0,0.5) ellipse (0.5 and 1);

 \node (B) at (2.8,0.5) {\normalsize B};
 \draw (2,0.5) ellipse (0.5 and 1);

 \draw[->] (a) to [bend left=20] node[above] {$2$} (b);
 \draw[->] (b) to [bend left=20] node[left] {$1$} (c);
 \draw[->] (c) to [bend left=20] node[below] {$2$} (d);
 \draw[->] (d) to [bend left=20] node[right] {$1$} (a);
 \draw[<-] (a) to [bend right=20] node[below] {$2$} (b);
 \draw[<-] (b) to [bend right=20] node[right] {$1$} (c);
 \draw[<-] (c) to [bend right=20] node[above] {$2$} (d);
 \draw[<-] (d) to [bend right=20] node[left] {$1$} (a);
\end{tikzpicture}

\end{center}
\caption{Left: An $\mathrm{OPT}/2$ selfish equilibrium. Right: An $\mathrm{OPT}/3$ tribal equilibrium.\vspace{1em}}
\label{fig:friendalt}
\end{figure}
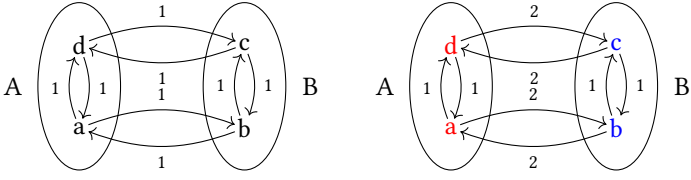
This circumstance does not change when players are fully altruistic, 
 and in fact, the lower bound is established by the following Nash equilibrium in the same example as seen for the preceding theorem.

In Fig. \ref{fig:friendalt} on the left,

we see that by switching to the other clique, each player would gain (from the unique player that benefits them in the other group) 1 unit of utility, and lose (from the unique player that benefits them in the current group) 1 as well. Also, the net loss to the rest of the community (as the person sharing the clique with the switching player would no longer benefit from their friendship) and the net gain (as one person in the clique they are switching to would now benefit) each are 1 as well, and so switching would indeed make no difference to the player's subjective utility. On the other end, we can establish a matching upper bound with little effort.

\begin{theorem} \label{thm:socialaltruism} The Price of Tribalism for the social $2$-grouping game and constant partition functions $\tau\in \T^{(1)}$ is 2.
\end{theorem}

What happens once different tribes enter the picture?
In the above example, the player's own loss

 due to lost friends was neatly cancelled by gained benefit due to newly gained friends, and likewise the loss to the community was cancelled by the gain experienced by the members of the defector's new peer group.

With multiple tribes, we no longer have to make a putative defector value the gains and losses of all other players equally.

How much worse could we make the equilibrium by making the player care about the friends he is currently in a clique with, but wouldn't care about the benefit he could bring to the members of the other clique?

The defector would have to value the benefit \emph{to himself} of the other group's friendship higher than the sum of the benefit to himself of his current group's \emph{and} the benefit his current group derives from him.
Therefore, the foregone friendship of those who are in the other clique could be worth up to twice as much to the player before he is incentivised to switch:
if we mark the members of one tribe {\color{red} red} and the members of the other tribe {\color{blue} blue}, the choice of cliques shown in
Fig. \ref{fig:friendalt} on the right

is pure Nash. Here, each player would gain 2 by defecting, and their tribe would gain 0; at the same time, they would lose 1, and their community would also lose 1, and so defecting is zero-sum. This example turns out to be tight for $2$-grouping, regardless of how many tribes we allow the players to belong to.

\begin{theorem} \label{thm:pot2grouping} The Price of Tribalism for the social $2$-grouping game and arbitrary partition functions $\tau \in \T^{(*)}$ is 3.
\end{theorem}

\begin{proof}
The above example shows a lower bound. For the upper bound, derive from the Nash condition summed over all players:
\begin{eqnarray*}
  \sum_i\sum_{j: \s(i)=\s(j)} u_{ji} + \sum_i \sum_{j: \substack{\s(i)=\s(j) \\ \tau(i)=\tau(j)}} u_{ij} &\geq& \sum_i \sum_{j: \s(i)\neq \s(j)} u_{ji} + \sum_i \sum_{j: \substack{\s(i)\neq\s(j) \\ \tau(i)=\tau(j)}} u_{ij} \\
 \sum_i\sum_{j: \s(i)=\s(j)} u_{ji} + \sum_i \sum_{j: \substack{\s(i)=\s(j) \\ \tau(i)=\tau(j)}} u_{ij} &\geq& \sum_i \sum_{j: \s(i)\neq \s(j)} u_{ji} \\
 2 \sum_i\sum_{j: \s(i)=\s(j)} u_{ji} &\overset{(*)}{\geq}& \sum_i \sum_{j: \s(i)\neq \s(j)} u_{ji}\\
 3 \sum_i\sum_{j: \s(i)=\s(j)} u_{ji} &\geq& \sum_i \sum_{j: \s(i)\neq \s(j)} u_{ji} + \sum_i \sum_{j: \s(i)=\s(j)} u_{ji} \\
 &=& \sum_i\sum_j u_{ji} = U(\s^*),
\end{eqnarray*}
where $(*)$ is as all summands are non-negative and relaxing the restriction only adds more terms: $$\sum_i\sum_{j: \s(i)=\s(j)} u_{ij}  \geq \sum_i \sum_{j: \substack{\s(i)=\s(j) \\ \tau(i)=\tau(j)}} u_{ij}.$$
As the sum on the LHS is just the social welfare at Nash, this completes the proof.
\end{proof}

The above argument can be extended to show that $k$ tribes are in fact always strictly worse than $k-1$ whenever there are at least $k$ distinct friendship cliques to be formed; see Theorem \ref{thm:potkgrouping}.

To the extent to which the social grouping game is a useful model of real-world social networks, we see this result as confirmation of an intuitively relatable phenomenon: when individuals ``fall in with the wrong crowd'', they can get stuck in local minima that are quite bad for everyone.

\subsection{Network contribution games}
\label{sec:contribution-games}

A more involved model of social relationships are the \emph{network contribution games} $\mathcal{N}$ described by Anshelevich and Hoefer \cite{contribution-games}. We model a social graph in which the vertices represent players who can divide up a personal \emph{budget} of effort $B_i \geq 0$ among their potential relationships, represented by edges. The benefit each player derives from a relationship $e=\{i,j\}$ is given by a non-negative, non-decreasing and symmetric \emph{reward function} $f_e: \R^2_{\geq 0} \to \R_{\geq 0}$ in terms of the amount of effort each of them invests, and each player's total utility is just the sum of benefits from all relationships.

In the original paper, players are both allowed to deviate individually if it benefits themselves and to coordinate a joint deviation with a neighbouring player if it benefits them both; this is both seen as more realistic for pairwise relationships, and necessary to enable interesting defection patterns for some classes of payoff functions to exist at all. We will follow this approach in the tribal extension, calling deviations by single players and connected pairs \emph{unilateral} and \emph{bilateral} respectively.
We will denote an equilibrium stable against bilateral deviations as a \emph{pairwise equilibrium}, giving rise to a \emph{pairwise Price of Tribalism}. Furthermore, we note that all lower-bound examples in this section work even when entire tribes are allowed to coordinate their deviation; we will call the corresponding PoT \emph{coordinated}.

\begin{definition} Some notation for the rest of this section.
\begin{enumerate}[label=(\roman*)]
\item Denote by $\s_i (e)$ the amount that player $i$ contributes into edge $e$ in strategy $\s$.
\item For each edge $e = \{i,j\}$, let $w_e (\s) = f_e(\s_i(e), \s_j (e))$ be the reward that the edge pays to both $i$ and $j$.
\item An edge $e = \{i,j\}$ is called \emph{tight} in strategy $\s$ if $\s_i(e) \in \{0,B_i\}$ and  $\s_j(e) \in \{0,B_j\}$, and a strategy $\s$ is tight if all edges are tight.
\end{enumerate}
\end{definition}

When the reward is
just
 a weighted sum of investments, \cite{contribution-games} shows that the PoA is 1. It is easy to establish that this is also the case when players are altruistic.

\begin{corollary}[of Theorem 2.8 in \cite{contribution-games}]
   $\PoT(\Tn{1},\mathcal{N}_{+})=1$.
  \label{cor:NSGsum1}
\end{corollary}

\begin{proof}
The socially optimal strategies are just those where each player puts all of their budget onto the edge adjacent to them with the largest reward. For each player, this is also the locally best strategy for maximising social welfare, since their contribution to social welfare is just two times the portion of the reward that is due to their own budget.
\end{proof}

However:

\begin{theorem}
\label{lem:NCGsum2}
  Suppose all reward functions are of the form 
 $c_e(x+y)$, where $c_e > 0$. Then the pure, pairwise and coordinated PoT each is $\PoT(\Tstar,\mathcal{N}_{+}) = 2$.
\end{theorem}

\begin{proof}
As noted in the proof of Theorem 2.8 in \cite{contribution-games}, the social optimum $\s^*$ is attained when all players invest their entire budget in the respective adjacent edge $e^*$ with maximum $c_e$. For a configuration $\s$ to be a Nash equilibrium, no player may want to deviate to investing their budget like this. If a player $i$ invests $\s_i(e)$ units into an adjacent edge $e\sim i$, their tribe earns $2\s_i(e)c_e$ units of utility if the player on the other end is also in the same tribe, and $\s_i(e)c_e$ units otherwise.
So $\sum_{e\sim i} 2 \s_i(e)c_e \geq B_i c_{e^*}$. Summing over all players, we find that $$U(\s) = \sum_i \sum_{e\sim i} \s_i(e)c_e \geq \frac{1}{2} \sum_i B_i c_{e^*(i)} = U(\s^*),$$ and so the PoT is bounded above by 2. \qedhere

\end{proof}

\begin{proof}[Proof (lower bound)]
A tight lower bound is given by the following graph:
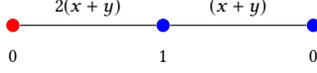
\begin{center}
\begin{tikzpicture}[every node/.append style={font=\scriptsize}]
\tikzstyle{s} = [circle, minimum width=5pt, fill, inner sep=0pt]
    \node[s,color=red] (x) at (0,0) {} node[below=0.1 of x] {$0$};
    \node[s,color=blue] (y) at (2,0) {} node[below=0.1 of y] {$1$};
    \node[s,color=blue] (z) at (4,0) {} node[below=0.1 of z] {$0$};

    \draw (x) -- node[above] {$2(x+y)$} (y) -- node[above] {$(x+y)$} (z);
\end{tikzpicture}
\end{center}
Here, the two players on the right are in the same tribe, but only the middle player has any budget. It would be socially optimal for them to invest this in the edge on the left, attaining a social welfare of $4$; however, the configuration where they instead invest in the edge on the right -- yielding a social welfare of $2$ -- is stable, since the blue tribe's total utility is $2$ regardless of how the middle player's budget is allocated. \qedhere
\end{proof}

When all reward functions are convex in each coordinate, \cite{contribution-games} shows a PoA of 2. In Theorem \ref{thm:NCGaltconvex}, we show
that this is also the case given altruism. Again, though, tribalism leads to deterioration.

\begin{theorem}
\label{lem:NCGconv4}
  Suppose all reward functions are convex in each coordinate. Then the pure, pairwise and coordinated Price of Tribalism each is $\PoT(\Tstar,\mathcal{N}_{C}) = 4$.
\end{theorem}

\begin{proof}[Proof (upper bound)]
  By Claim 2.10 in \cite{contribution-games}, since all reward functions are coordinate convex, we can assume that the optimum $\s^*$ is tight. Fix a pairwise tribal Nash equilibrium $\s$. Note that we can normalise the $f_e$'s so that $f_e(0,0)=0$. Since the reward functions are non-decreasing, the normalised functions will still be valid reward functions, and subtracting a constant from utility at both OPT and Nash can only increase their ratio.

In a tight strategy, each player can invest their budget in at most one edge. When a player $i$ invests in an edge $e$ in the optimum solution $\s^*$, we will say that $i$ is a \emph{witness} to $e$.
 Let $e = \{i,j\}$ be an edge where $\s^*_i(e) = B_i$ and $\s^*_j (e) = B_j$, so $i$ and $j$ are both witnesses to $e$. By the Nash condition, if $i$ and $j$ were to bilaterally deviate to their strategies in $\s^*$, then it must not be beneficial for at least one of the two players' tribes. Suppose WLOG that this is $i$.
  In the worst case, $i$ and $j$ were in the same tribe \emph{and} benefitting other members of their tribe, and so the tribe loses $2(u_i(\s)+u_j(\s))$. On the other hand, the worst-case gain occurs when $i$ and $j$ are in different tribes, and so the switch only benefits $i$'s tribe

 one lot of $w_e(\s^*)$. So by the Nash condition, we can derive
  \begin{align*}
    u_i^\tau (\s) \geq u_i^\tau (\s_i^*; \s_j^*; \s_{-i,j}) \geq u_i^\tau (\s) - 2 (u_i(\s) + u_j(\s)) + w_e (\s^*).
  \end{align*}
Rearranging the inequality, we have $2 (u_i(\s) + u_j(\s)) \geq w_e(\s^*)$. So $w_e(\s^*)$ is less than two times the sum of the utilities of its witnesses in $\s$.
So suppose instead $e = \{i,j\}$ is an edge where $\s^*_i(e) = B_i$ and $\s^*_j (e) = 0$, so only $i$ is a witness to $e$.
By the same reasoning as above, we have
  \begin{align*}
    u_i^\tau (\s) \geq u_i^\tau (\s_i^*; \s_{-i}) \geq u_i^\tau(\s) - 2 u_i (\s) + w_e (\s^*).
  \end{align*}
  So again, $w_e(\s^*)$ is less than two times the sum of the utilities of its witnesses in $\s$.

Since each player is marked as a witness to exactly one edge, we can sum the above inequalities, treating one side as a sum over all edges and the other as a sum over all players. We thus conclude
  \begin{align*}
    U(\s^*) = 2 \sum_e w_e(\s^*) \leq 4 \sum_{i \in V} u_i(\s) = 4 U(\s). \tag*{\qedhere}
  \end{align*}
\end{proof}

\begin{proof}[Proof (lower bound)]
The following example in fact provides a matching lower bound for any function class that contains a coordinate convex function $f$ satisfying $f(x,0)=0$ and is closed under scalar multiplication.
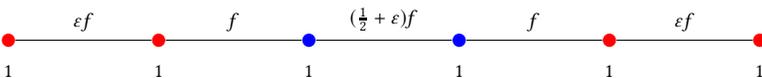
\begin{center}
\begin{tikzpicture}[every node/.append style={font=\scriptsize}]
\tikzstyle{s} = [circle, minimum width=5pt, fill, inner sep=0pt]
    \node[s,color=red] (n1) at (0,0) {};
    \node[s,color=red] (n2) at (2,0) {};
    \node[s,color=blue] (n3) at (4,0) {};
    \node[s,color=blue] (n4) at (6,0) {};
    \node[s,color=red] (n5) at (8,0) {};
    \node[s,color=red] (n6) at (10,0) {};
    \foreach \n in {1,...,6}{ \node[below=0.1 of n\n] {$1$}; }

    \draw (n1) -- node[above] {$\varepsilon f$} (n2) -- node[above]{$f$} (n3);
    \draw (n3) -- node[above] {$(\frac12+\varepsilon)f$} (n4);
    \draw (n6) -- node[above] {$\varepsilon f$} (n5) -- node[above]{$f$} (n4);
\end{tikzpicture}
\end{center}
The social optimum, with welfare $4f(1,1)$, is attained when the four players in the middle invest their budgets in the respective adjacent edge with payoff $f$. However, we can show that the configuration in which all budget is invested in the first, third and fifth edge, for a total payoff of $(2\varepsilon+1+2\varepsilon+2\varepsilon)f(1,1)$, is stable against unilateral, bilateral and whole-tribe deviations:
No set of players who are in the same tribe will want to deviate, as
this would involve diverting budget from an edge that has investments on both ends (thus losing utility) to one that has no investment on the other end (thus not gaining any).
Also, the two (distinct-tribe) players at the second and fourth edge in the graph will not want to deviate together, because this will not benefit
the blue tribe player closer to the center: supposing they divert $b$ units and the red player diverts $a$ units to their shared edge, we have
$$ f(a,b) + 2(1/2+\varepsilon) f(1,1-b) < 2(1/2+\varepsilon) (f(1,b)+f(1,1-b)) < 2(1/2+\varepsilon) f(1,1) $$
by non-decreasingness and coordinate convexity.
\end{proof}

\subsection{Atomic linear routing games}
\label{sect:atomiclinear}

Atomic linear routing games were first defined in \cite{rose-1973}, and their prices of anarchy were first studied in \cite{og-atomicroutingPoA} in the context of asymmetric scheduling games; an exposition of this is given in \cite{r-20lectures}. In these games, each player $i$ is associated with a pair of vertices $(s_i, t_i)$ of a directed graph, called its \emph{source} and \emph{sink} respectively. We think of the game as modelling multiple players traversing a road network, incurring some delays along the way depending on the total congestion on each road segment traversed. In the linear case, these delay functions are assumed to be linear, so each edge $e$ is associated with a positive factor $\alpha_e$ such that when $k$ players are on the edge \ycomment{$k$ is normally used for the number of tribes}, each of them incurs a delay of $\alpha_e k$ (and hence the sum of their delays is $\alpha_e k^2$). Formally, the strategies available to player $i$ are the set of paths from $s_i$ to $t_i$ in the graph, and the cost incurred by the player is $c_i(\s)=\sum_{e\in \s_i} \alpha_e \#\{j:e\in\s_j\}.$

By \cite{ChristodoulouK05,a-atomic-routing}, the pure Price of Anarchy for atomic linear routing games is exactly $\frac{5}{2}$.
In \cite{cara}, a weaker upper bound of $3$ is shown to hold if players are at least partially altruistic (optimising some convex combination of their own utility and social welfare), and this bound is tight when players are fully altruistic. We will demonstrate that this bound does not hold when players show tribal altruism towards  two or more tribes.
The example that gives rise to our lower bound relies on tribal behaviour that is quite intuitive: at certain interior nodes (case \ref{enum:case1} below), a tribally altruistic player prefers to continue paying a greater cost (while also causing a great cost to an ``outgroup'' member) over switching to a configuration that would benefit both the player and the commons, but result in a greater cost being paid by the player's tribe.

The matching upper bound uses smoothness.

\begin{theorem} \label{thm:alrg3k} The Price of Tribalism for atomic linear routing games $\mathcal{R}$ with 2-tribe partition functions $\tau \in \Tn{2}$, as well as arbitrary partition functions in $\Tstar$, is

$$ \PoT(\Tn{2},\mathcal{R}) = \PoT(\Tn{*},\mathcal{R}) = 4. $$
\end{theorem}

\begin{proof}[Proof (lower bound)]
Our construction is inspired by the construction in \cite{cara}.

As in that paper, our example will be formulated not as a routing game, but as a specific load-balancing game in which each player (represented as an edge) can choose between one of exactly two ``servers'' or congestible elements with linear cost functions (represented as the endpoints of the edge). This representation can be converted back into a routing game by the following scheme:
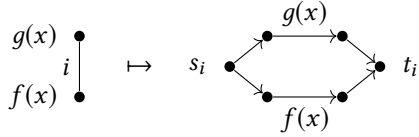
\begin{center}
\begin{tikzpicture}
\tikzstyle{subj} = [circle, minimum width=4pt, fill, inner sep=0pt]
    \begin{scope}[yscale=0.8]
    \node[subj] (f) at (0,0) {} node[left=0.1 of f] {$f(x)$};
    \node[subj] (g) at (0,1) {} node[left=0.1 of g] {$g(x)$};
    \draw (f) -- node[left] {$i$} (g);

    \node (lol) at (0.8,0.5) {$\mapsto$};

    \node[subj] (s) at (2,0.5) {} node[left=0.1 of s] {$s_i$};
    \node[subj] (a) at (2.5,0) {};
    \node[subj] (b) at (3.5,0) {};
    \node[subj] (c) at (2.5,1) {};
    \node[subj] (d) at (3.5,1) {};
    \node[subj] (t) at (4,0.5) {} node[right=0.1 of t] {$t_i$};
    \draw[->] (s)->(a); \draw[->] (a)->node[below] {$f(x)$} (b); \draw[->] (b)->(t);
    \draw[->] (s)->(c); \draw[->] (c)->node[above] {$g(x)$} (d); \draw[->] (d)->(t);
    \end{scope}
\end{tikzpicture}
\end{center}
For every $k$, we will now construct a game $G_k$ and describe a tribal Nash $\vs^k$. 

The game is played on a binary tree with $k+1$ layers of nodes (and hence $k$ layers of edges). Unlike the construction of \cite{cara} (Thm. 2), we do not require to introduce additional
edges below the tree, since the costs in the layers of our construction decay fast enough that the total weight of the final layer is dominated by the rest of the tree.

We set the delay function of the nodes at depths (distances from the root) $i=0,1,\ldots,k-1$ to be $f_i(x)=\left(1/2\right)^i x.$
The cost of the nodes in the final layer shall instead be \emph{twice} that of the preceding layer: $f_k(x)=\left(1/2\right)^{k-1} \cdot 2\cdot x.$

Each of the two players (edges) under a node shall belong to different tribes, say the left edge to tribe {\color{red} 1} and the right edge to tribe {\color{blue} 2}.

The overall construction will then look like this:
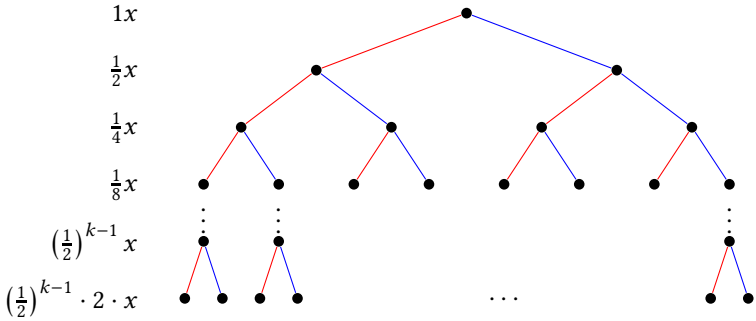
\begin{center}
\begin{tikzpicture}
\tikzstyle{s} = [circle, minimum width=4pt, fill, inner sep=0pt]
    \begin{scope}[xscale=0.5,yscale=0.75]
    \node[s] (r) at (0,0) {};
    \node[s] (n10) at (-4,-1) {}; \node[s] (n11) at (4,-1) {};
    \node[s] (n20) at (-6,-2) {}; \node[s] (n21) at (-2,-2) {}; \node[s] (n22) at (2,-2) {}; \node[s] (n23) at (6,-2) {};
    \node[s] (n30) at (-7,-3) {}; \node[s] (n31) at (-5,-3) {}; \node[s] (n32) at (-3,-3) {}; \node[s] (n33) at (-1,-3) {};
    \node[s] (n34) at (1,-3) {}; \node[s] (n35) at (3,-3) {}; \node[s] (n36) at (5,-3) {}; \node[s] (n37) at (7,-3) {};

    \draw[color=red] (r)--(n10)--(n20)--(n30);
    \draw[color=red] (n21)--(n32); \draw[color=blue] (n10)--(n21);
    \draw[color=red] (n22)--(n34);
    \draw[color=red] (n23)--(n36);
    \draw[color=blue]  (r)--(n11)--(n23)--(n37);
    \draw[color=blue] (n22)--(n35); \draw[color=red] (n11)--(n22);
    \draw[color=blue] (n20)--(n31);
    \draw[color=blue] (n21)--(n33);

    \node (dots) at (-7,-3.5) {$\vdots$};
    \node (dots) at (-5,-3.5) {$\vdots$};
    \node (dots) at (7,-3.5) {$\vdots$};

    \node[s] (n40) at (-7,-4) {}; \node[s] (n41) at (-5,-4) {}; \node[s] (n42) at (7,-4) {};
    \node[s] (n50) at (-7.5,-5) {}; \node[s] (n51) at (-6.5,-5) {}; \node[s] (n52) at (-5.5,-5) {}; \node[s] (n53) at (-4.5,-5) {}; \node[s] (n54) at (6.5,-5) {}; \node[s] (n55) at (7.5,-5) {};
    \draw[color=red] (n40)--(n50);
    \draw[color=red] (n41)--(n52);
    \draw[color=red] (n42)--(n54);
    \draw[color=blue] (n40)--(n51);
    \draw[color=blue] (n41)--(n53);
    \draw[color=blue] (n42)--(n55);

    \node (dots) at (1,-5) {$\hdots$};

    \tikzstyle{t} = [anchor=east]
    \node[t] (lol) at (-8.5, 0) {$1x$};
    \node[t] (lol) at (-8.5,-1) {$\frac{1}{2}x$};
    \node[t] (lol) at (-8.5,-2) {$\frac{1}{4}x$};
    \node[t] (lol) at (-8.5,-3) {$\frac{1}{8}x$};
    \node[t] (lol) at (-8.5,-4) {$\left(\frac{1}{2}\right)^{k-1}x$};
    \node[t] (lol) at (-8.5,-5) {$\left(\frac{1}{2}\right)^{k-1}\cdot 2\cdot x$};

    \end{scope}

\end{tikzpicture}
\end{center}
We claim that the strategy profile $\vs^k$ in which every player-edge chooses to occupy the ``upper'' (closer to the root) vertex is Nash.

Indeed, by analysing the environment of each edge depending on the layer it is situated in, we can verify the Nash condition for all players.

\begin{enumerate}
\item \label{enum:case1} Intermediate layers, up to exchange of tribes:

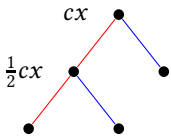
\begin{tabular}{ p{7em} p{23em} }
\begin{tikzpicture}[baseline={($(b.base)+(0,0.2)$)}]
\tikzstyle{s} = [circle, minimum width=4pt, fill, inner sep=0pt]
    \begin{scope}[xscale=0.3,yscale=0.75]
    \node (b) at (0,0.3) {};
    \node[s] (r) at (-2,0) {} node[left=0.2 of r] {$cx$} ;
    \node[s] (n10) at (-4,-1) {} node[left=0.2 of n10] {$\frac{1}{2}cx$}; \node[s] (n11) at (0,-1) {};
    \node[s] (n20) at (-6,-2) {}; \node[s] (n21) at (-2,-2) {};

    \draw[color=red] (r)--(n10)--(n20);
    \draw[color=blue] (n10)--(n21);
    \draw[color=blue]  (r)--(n11);
    \end{scope}
\end{tikzpicture}
 & \vspace{0.2em} In this case, at Nash, the top red player incurs a cost of $1c\cdot 2 + \frac{1}{2}c 2 = 3c$ (for himself on the two-player node above, and his tribesman on the two-player node below). If he were to switch down, his cost would be $0 + \frac{1}{2} c 3\cdot 2$%

, which is also $3c$.
\end{tabular}
\vspace{0.5em}
\item Final layer, up to exchange of tribes:

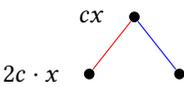
\begin{tabular}{ p{7em} p{23em} }
\begin{tikzpicture}[baseline={($(b.base)+(0,0.5)$)}]
\tikzstyle{s} = [circle, minimum width=4pt, fill, inner sep=0pt]
    \begin{scope}[xscale=0.3,yscale=0.75]
    \node(b) at (0,0.3) {};
    \node[s] (r) at (-2,0) {} node[left=0.2 of r] {$cx$} ;
    \node[s] (n10) at (-4,-1) {} node[left=0.2 of n10] {$2c \cdot x$}; \node[s] (n11) at (0,-1) {};

    \draw[color=red] (r)--(n10);
    \draw[color=blue] (r)--(n11);
    \end{scope}
\end{tikzpicture}
 & \vspace{0.2em} 
In this case, at Nash, the red player incurs a cost of $2c$, as nobody is using the bottom node and he is sharing the top node with a player from the other tribe.
If he were to switch down, he would be using the node alone, but his cost would still be $2c$.
\end{tabular}
\end{enumerate}

Summing by layer, the total cost of this assignment then is
\begin{eqnarray*}
C_{G_k}(\vs^k) &=& \sum_{i=0}^{k-1}  {\color{red} 4} \cdot {\color{blue} 2^i} \cdot \left( \frac{1}{2} \right)^i = 4k. 
\end{eqnarray*}
Here, the cost factor due to congestion on each vertex is {\color{red} red}, and the number of vertices in each layer is {\color{blue} blue}. The cost factor on each vertex is black.

On the other hand, the social optimum is at least as good as the strategy $(\vs^k)^*$ where every player uses the node further ``down'' (away from the root).
In this assignment, every vertex except for the root is occupied by exactly one player, so the cost of the optimum is bounded above by the total cost

\begin{eqnarray*}
C_{G_k}(\text{opt}) \leq C_{G_k}((\vs^k)^*) &= & \sum_{i=0}^{k-1}  {\color{red} 1} \cdot {\color{blue} 2^i} \cdot \left( \frac{1}{2} \right)^i \underbrace{-1}_{\text{root}}
+ \underbrace{ {\color{red} 1} \cdot {\color{blue} 2^k} \left(\frac{1}{2}\right)^{k-1} \cdot 2}_{\text{bottom-most row}} \\
&=& k-1+4.
\end{eqnarray*}
Hence we can conclude that as $k\rightarrow\infty$, the ratio between the cost of the Nash equilibrium and the social optimum goes to $4$ from below: that is, for any $\varepsilon$, there is a $k$ such that
$$ \frac{ C_{G_k}(\vs^k) }{ C_{G_k}( \text{opt} ) } \geq  \frac{4k}{k+3} \geq 4 - \varepsilon $$
as claimed. \ycomment{there are two periods}
\end{proof}

In order to establish the upper bound (which holds for any number of tribes), we will first need to introduce an appropriate instance of the common notion of \emph{smoothness}, originally due to Roughgarden \cite{smooth-roughgarden}. Broadly speaking, a smooth game is one in which in expectation, a unilateral deviation towards a different strategy profile moves the deviating player's welfare towards some multiple of its welfare in the target profile. This property can be used to deduce a generic bound on the Price of Anarchy.
\begin{definition}
  Let $G^\tau$ be the tribal extension of a finite cost-minimisation game $G$. $G$ is \emph{$(\lambda,\mu,\tau)$-smooth} if for any strategy profiles $\vs, \vs' \in \Sigma$,
        $$ \sum_{i \in N} \left( c_i^\tau (\vs_i'; \vs) - (c_i^\tau(\vs) - c_i(\vs)) \right) \leq \lambda C(\vs') + \mu C(\vs).$$
  \label{defn:tribalsmoothness}
\end{definition}

Other work in the literature on altruism and social context uses generalisations of smoothness. Of particular note is Chen's notion of $(\mu,\lambda,\alpha)$-altruistic smoothness \cite{chen} and Rahn and Sch\"afer's $\mathcal{SC}$-smoothness \cite{SCsmoothness}. Our definition agrees with Roughgarden's when $\tau$ assigns each player to his own tribe, and with $(\mu, \lambda, \textbf{1})$-altruistic smoothness when $\tau$ assigns all players to the same tribe;
it also turns out that $\mathcal{SC}$-smoothness is a straightforward generalisation.

\begin{theorem}[\cite{SCsmoothness}]
 Let $\mathcal{G}$ be a class of games and $\T = \{ \T_G \}_{G \in \mathcal{G}}$ be a class of partition functions for each game. If for every $G \in \mathcal{G}$ and $\tau \in \T_G$, $G^\tau$ is $(\lambda,\mu,\tau)$-smooth, then
$\PoT( \T, \mathcal{G} ) \leq \lambda/(1-\mu).$
  \label{thm:smooth}
\end{theorem}
\begin{proof}
Fix a $G \in \mathcal{G}$ and $\tau \in \T_G$. Let $\vs$ be a tribal Nash equilibrium of game $G^\tau$.
Then
\begin{eqnarray*}
C(\vs) &=& \sum_{i \in N} c_i(\vs) = \sum_{i \in N} (c_i^\tau(\vs) - c_i^\tau(\vs) + c_i(\vs)) \\
&\leq& \sum_{i \in N} c_i^\tau (\vs^*_i; \vs_{-i})  - (c_i^\tau(\vs) - c_i(\vs)) \\
&\leq& \lambda C(\vs^*) + \mu C(\vs),
\end{eqnarray*}
where the first inequality follows from the tribal Nash condition. Since this is true for all games and partitions, we have $\PoT \leq \lambda/ (1 - \mu)$.
\end{proof}

The following bound, which is in the spirit of several similar ones in the literature (e.g. \cite{chen} Lemma 4.4), will be a key ingredient
in the proof to follow.
\begin{lemma} For integers $x,y\geq 0$, $ x(y-x) + xy + x + y \leq \frac83 y^2 + \frac13 x^2. $
\label{lem:quad}
\end{lemma}
\begin{proof}
We will equivalently show that
$$ 2xy+x+y \leq \frac83 y^2 + \frac43 x^2. $$
Let $f(x,y)=2xy+x+y$, $g(x,y)=\frac83 y^2 + \frac43 x^2$. We have
$$ f(x,0)=x \text{ and } f(0,y)=y; $$
also,
$$ g(x,0)=\frac43 x^2 \text{ and } g(0,y)=\frac83 y^2, $$
and so it is easy to verify that for all such pairs, $f(x,y)\leq g(x,y)$ as required.
Using this as the base case, observe now that
\begin{eqnarray*}
& & f(x+1,y+1)-f(x,y)= 2x+2y+4 \\
&\leq& g(x+1,y+1)-g(x,y)=\frac83 (2y+1) + \frac43 (2x+1) = \frac{16}3y+\frac{8}3x + 4
\end{eqnarray*}
for all $x,y\geq 0$. So if $f(x,y)\leq g(x,y)$, then $f(x+1,y+1)\leq g(x+1,y+1)$.
Thus, we can derive the inequality for all $(x,y)\in \mathbb{N}^2$ inductively.
\end{proof}

\begin{lemma}
  Let $G^\tau$ be a $\tau$-tribal extension of an atomic linear routing game. Then $G^\tau$ is ($8/3,1/3,\tau$)-smooth.
\label{lem:alrgsmooth}
\end{lemma}
\begin{proof}
  Let $\vs$ and $\vs^*$ be two strategy profiles in any $\tau$-extension of any atomic linear congestion game.
  We will use
  $ n_e(\s) = \#\{i \mid e\in \s_i\} $
  to denote the number of players using edge $e$ in strategy $\s$, and
  $ n_e^t(\s) = \#\{i \mid e\in \s_i, \tau(i)=t\} $
  be the number of players on edge $e$ that belong to tribe $t$.
  Then $ c_i^\tau( \vs) = \sum_e \alpha_e n_e^{\tau(i)} (\vs) n_e(\vs). $
  For each player $i$, we can compute the change in cost of $i$'s tribe as she switches from $\vs$ to $\vs^*$,
  \begin{align*}
    c_i^\tau & (\vs_i^*; \vs_{-i}) - c_i^\tau(\vs)
     = \sum_{e \in \vs_i^* \setminus \vs_i} \alpha_e ( (n_e^{\tau(i)} (\vs) + 1) ( n_e (\vs) + 1) - n_e^{\tau(i)} (\vs) n_e (\vs) ) \\
       &+ \sum_{e \in  \vs_i \setminus \vs_i^*} \alpha_e ( (n_e^{\tau(i)} (\vs) - 1) ( n_e (\vs) - 1) - n_e^{\tau(i)} (\vs) n_e (\vs) ) \\
    &\leq \sum_{e \in \vs_i*} \alpha_e (n_e^{\tau(i)} (\vs) + n_e (\vs) + 1) + \sum_{e \in \vs_i} \alpha_e ( 1 - n_e^{\tau(i)} (\vs) - n_e (\vs) ).
  \end{align*}
  Here, the last inequality is because we can add the (always positive) contribution of edges $e\in \vs_i \cap \vs_i^*$.
  Then, substituting into the left hand side of Definition \ref{defn:tribalsmoothness} and using that $c_i(\vs)=\sum_{e\in \vs_i} \alpha_e n_e(\vs)$, we find that
\begin{eqnarray*}
& & \sum_{i\in N} (c_i^\tau (\vs_i^*; \vs_{-i}) - c_i^\tau(\vs) + c_i(\vs)) \\
&\leq& \sum_{\text{tribes }t} \sum_{i \in N: \tau(i)=t} \left( \sum_{e \in \vs_i^*} \alpha_e (n_e^t (\vs) + n_e (\vs) + 1) + \sum_{e \in \vs_i} \alpha_e ( 1 - n_e^t (\vs)) \right) \\
&=& \sum_{\text{tribes }t} \sum_{\text{edges }e} \alpha_e \left( n_e^t(\vs^*) (n_e^t(\vs) + n_e(\vs)+1) + n_e^t(\vs) ( 1-n_e^t(\vs) ) \right)
\end{eqnarray*}
by changing the order of summation and combining the $n_e^t(\vs^*)$ (resp. $n_e^t(\vs)$) identical summands on each edge; this is
\begin{eqnarray*}
&=& \sum_{\text{tribes }t} \sum_{\text{edges }e} \alpha_e \left( n_e^t(\vs) ( n_e^t(\vs^*) - n_e^t(\vs) ) + n_e^t(\vs^*) n_e(\vs) + n_e^t(\vs^*) + n_e^t(\vs)  \right) \\
&\leq& \sum_{\text{edges }e} \alpha_e \left( n_e(\vs) ( n_e(\vs^*)-n_e(\vs) ) + n_e(\vs^*) n_e(\vs) + n_e(\vs^*) + n_e(\vs)  \right)
\end{eqnarray*}
by summing over tribes and using $n_e^t(\vs)\leq n_e(\vs)$ (as the tribes are a partition of all players using edge).
By Lemma \ref{lem:quad}, we conclude that this is
\begin{eqnarray*}
&\leq& \sum_{\text{edges }e} \alpha_e \left( \frac83 n_e(\vs^*)^2 + \frac 13 n_e(\vs)^2) \right) = \frac83 C(\vs^*) + \frac13 C(\vs). \qedhere
\end{eqnarray*}
\end{proof}

\begin{proof}[Proof (upper bound of Thm. \ref{thm:alrg3k})]
Follows from Lemma \ref{lem:alrgsmooth} and Thm. \ref{thm:smooth}.
\end{proof}

The proof goes through unchanged for general atomic congestion games with linear costs.

\section{Conclusions}
We have introduced the concept of a $\tau$-tribal extension of a given strategic game $G$
and used it to define the Price of Tribalism,
a measure of the badness of possible
equilibria when players are altruistic among each other within groups and indifferent
towards members of other groups. By means of three different examples, we have demonstrated
that in many cases, the Price of Tribalism can be worse than either the corresponding
measure when players are completely selfish or when they are completely altruistic.
While we have mostly focused on scenarios in which the members of a tribe
cannot coordinate their actions, our analysis of the network contribution game
shows that this can still be the case even when full coordination at the tribal
level is in fact possible.

While all examples giving rise to our lower bounds were based on the idea that a player would be willing
to forego a deviation beneficial to themselves and the commons for the local benefit of their tribe,
it remains an open question how generally this intuition can give rise to lower bounds on Price of Tribalism
that beat known upper bounds on the pure Price of Anarchy. For some games, such as the opinion-forming game
of \cite{bindel}, we suspect that even though tribalism gives rise to new equilibria, those are never worse
than the ones that arise from selfishness. Future work will focus on investigating other games and  identifying
 general conditions that are sufficient for the Price of Tribalism to exceed the Price of Anarchy.

\section*{Acknowledgements}

We would like to thank Jerry Anunrojwong, Ioannis Caragiannis, Artur Gorokh, Bart de Keijzer, Bobby Kleinberg, Guido Schaefer and {\'{E}}va Tardos, as well as the
anonymous reviewers, for helpful feedback and discussions regarding the paper. M.S. was supported by NSF grants IIS-1703846 and
IIS-1718108, ARO grant W911NF-17-1-0592, and a grant from the Open Philanthropy project.

\bibliographystyle{alpha}
\bibliography{tribib}

\clearpage

\appendix

\section{General miscellanea}
\label{sec:misc}

The class of equilibria we consider in the paper is an analogue of the standard notion of Nash equilibrium:
\begin{definition}
  Given a finite cost minimisation game $G=(\P, (\Sigma_i)_{i \in \P}, (c_i)_{i\in \P})$, $\vs \in \Sigma$ is a \emph{pure Nash equilibrium} if for all $i \in \P$ and $\vs' \in \Sigma$,
  $$c_i (\vs) \leq c_i (\vs'_i;\vs_{-i}).$$
\end{definition}
Adapting this definition for the tribal extension simply requires us to change the cost function that the players are now minimising.
\begin{definition}
  Given a $\tau$-tribal extension of a finite cost minimisation game $G=(\P, (\Sigma_i)_{i \in \P}, (c_i)_{i\in \P})$, $\vs \in \Sigma$ is a (pure) \emph{tribal Nash equilibrium} if for all $i \in \P$ and $\vs' \in \Sigma$,
  $$c_i^\tau (\vs) \leq c_i^\tau (\vs'_i;\vs_{-i}).$$
  \label{defn:tribalnash}
\end{definition}

One can define tribal mixed equilibrium, correlated equilibrium, and coarse correlated equilibrium in exactly the same way. In this paper, we will focus on
pure equilibria, which already exhibit interesting differences in the presence of tribes.

The Price of Anarchy \cite{anarchy} is a standard notion for measuring the inefficiency of selfish behaviour for classes of games.
\begin{definition}
Let $\mathcal{G}$ be a class of finite cost-minimisation games with social cost functions $C_G$ for each $G \in \mathcal{G}$.
The \emph{Price of Anarchy} of $\mathcal{G}$ is
$$\PoA (\mathcal{G})= \sup_{G \in \mathcal{G}} \frac{\sup_{\vs \in S_G} C_G(\s)}{C_G(\vs^*)}, $$
where $S_G$ is the set of pure Nash equilibria of $G$.
\end{definition}
This definition is frequently made with the set of strategy profiles $S_G$ to be compared against the optimum as an explicit parameter,
in which case the standard Price of Anarchy is obtained with the appropriate set of pure Nash equilibria of $G$, whereas
the Price of Tribalism is obtained when the set of all tribal equilibria under all possible partition functions, $\bigcup_{\tau\in \T} S_{G^\tau}$, is used
instead.

We can note some basic observations about the PoT that follow immediately from its definition.

\begin{proposition} (Basic facts)
\begin{enumerate}
\item If $\T_G \subseteq \T'_G$ for all $G\in \G$, then $\PoT(\T, \G) \leq \PoT(\T',\G)$. Hence $\mathrm{PoA}(\G) \leq \PoT(\Tstar,\G)$.
\item If the functions $\tau\in\T_G$ map each player to a distinct tribe for all $G$ (i.e. $i\neq j \Rightarrow \tau(i)\neq \tau(j)$), then the resulting tribal extensions are equivalent to the original game, and so $\mathrm{PoA}(\G) = \PoT(\T,\G)$.
\item If the class of games $\G$ is closed under some means of adding players whose cost/utility is always zero, then for any $k< k'$, any NE in a tribal extension with $k$ tribes corresponds to an NE of same value in a game with $k'$ tribes where the additional $k'-k$ tribes consist of zero players, and hence $\PoT(\Tn{k},\G) \leq \PoT(\Tn{k'},\G)$.

\end{enumerate}
\label{prop:basic}
\end{proposition}

\section{Missing proofs}

\subsection{About social grouping games}
\label{sec:a.social}

\ycomment{move statement to main paper?}
\begin{theorem}
  The selfish and altruistic Price of Anarchy of $k$-grouping games are both $k$.
  \label{thm:friendshipk}
\end{theorem}

\begin{proof}
  In this proof, we will use the notation $\s(i) = c$ for player $i$ choosing clique $c$ in strategy $\s$, and any unspecified sums over $c$ or $c'$ will be over all cliques and $i,j$ over all players.
  Recall that the social welfare for a strategy $\s$ is
  \[ U(\s) =  \sum_c \sum_{i,j: \s(i) = \s(j) = c} u_{ji}, \]
  and the social welfare of the socially optimum strategy (all players choosing the same clique) is
  \[ U(\s^*) =  \sum_{i,j} u_{ji}. \]

  First, we will establish an upper bound on the selfish Price of Anarchy. Let $\s$ be a Nash equilibrium of a $k$-grouping game and fix a clique $c$. The Nash condition says that for all players $i$ such that $\s(i) = c$ and all cliques $c'$,
  \[ \sum_{j: \s(j) = c} u_{ji} \geq \sum_{j: \s(j) = c'} u_{ji}. \]
Summing over all cliques $c'$ and players $i$ such that $\s(i) = c$, we get
  \begin{align*}
    k \sum_{i: \s(i) = c} \sum_{j: \s(j) = c} u_{ji} & \geq \sum_{i \in c} \sum_{c'}  \sum_{j: \s(j) = c'} u_{ji} \\
    k \sum_{i, j: \s(i) = \s(j) = c} u_{ji} & \geq \sum_{i \in c} \sum_{j} u_{ji} \\
    k \sum_c \sum_{i, j: \s(i) = \s(j) = c} u_{ji} & \geq \sum_{i,j} u_{ji} \qquad \text{(by summing over cliques $c$).}
  \end{align*}
  This gives us that the Price of Anarchy is less than or equal to $k$.

  The upper bound for altruistic PoA is very similar. The Nash condition for player $i$ deviating from clique $c$ to clique $c'$ is
  \[  \sum_{j: \s(j) = c} (u_{ij} + u_{ji}) \geq \sum_{j: \s(j) = c'} (u_{ij} + u_{ji}), \]
  since player $i$ now also cares about her effect on the other players, but her deviation does not affect players not in cliques $c$ or $c'$. As above, sum over $c'$, players $i$ in clique $c$, then $c$, and obtain:
   \begin{align*}
     k \sum_c \sum_{i, j: \s(i) = \s(j) = c} (u_{ij} + u_{ji}) & \geq \sum_{i,j} (u_{ij} + u_{ji}) \\
     2k \sum_c \sum_{i, j: \s(i) = \s(j) = c} u_{ji} & \geq 2 \sum_{i,j} u_{ji}.
   \end{align*}
   This gives us an upper bound of $k$ once again.

   For the lower bound consider the following game. Label the $2k$ players as $\{a_i\}_{i=1}^k \cup \{b_i\}_{i=1}^k$ and define the utilities (edge weights) to be
   \begin{align*}
     u_{a_i b_i} = u_{b_i a_i} &= 1 \quad \text{for all $i$} \\
     u_{a_i a_j} = u_{b_i b_j} &= 1 \quad \text{for all $i \neq j$} \\
     u_{xy} &=  0         \quad \text{otherwise}.
   \end{align*}
   Then consider a strategy $\s$ where the cliques are $\{a_i, b_i\}$ for all $i$. This is a Nash equilibrium: when a player $p$ deviate to a different clique, they would lose one unit of utility from their old clique and gain one from their new clique. This is also an altruistic Nash: when a player deviates, they lose two units of social welfare from the old clique and gain two in social welfare from the new clique. Computing the social welfare of $\s$ and the social optimum, $\s^*$, we get
   \[U(\s) = 2k \qquad U(\s^*) = 2k(k-1) + 2k = 2k^2. \]
   So the both the Price of Anarchy and the altruistic Price of Anarchy are $k$.
\end{proof}

\ycomment{move statement to main paper?} 
\begin{theorem} \label{thm:potkgrouping} The Price of Tribalism for social $k$-grouping games is exactly \[\PoT(\Tstar, \mathcal{F}_k) = 2k - 1.\]
\end{theorem} 
\begin{proof} (Sketch)
For the upper bound, the RHS of the first line in the proof of Theorem \ref{thm:pot2grouping} becomes a maximum over the $k-1$ possible alternative cliques that the $i$th player could join. So if we multiply the
inequality by $k-1$, the RHS is an upper bound on $$\sum_{c \neq \s(i)} \sum_{j:\s(j)=c} u_{ji} = \sum_{j:\s(j)\neq \s(i)} u_{ji}.$$ The rest of the proof proceeds identically except for the additional factor of $k-1$ on the LHS up to the line marked $(*)$, resulting in a factor of $2k-1$ rather than 3 on the final line.

\ycomment{I changed this}
For the lower bound, use the same construction as the one from Theorem \ref{thm:friendshipk} but with tribes $\{a_i, b_i\}$, for all $i$, and
   \begin{align*}
     u_{a_i a_j} = u_{b_i b_j} &= 2 \quad \text{for all $i \neq j$}.
   \end{align*}
   This strategy profile, $\s$, is a tribal Nash equilibrium since a player's tribe would lose utility 4 from her leaving a clique, but gain 4 from her joining any other. The social welfare of $\s$ compared to that of the optimum, $\s^*$, is
   \[ U(\s) = 2k \qquad U(\s^*) = 4k(k-1) + 2k = 4k^2 - 2k. \]
   This yields a Price of Tribalism of $2k-1$.
\end{proof}

\subsection{About network contribution games }
\label{sec:network}

\begin{theorem}
  \label{thm:NCGaltconvex}
  Suppose all reward functions are coordinate convex. Then the altruistic Price of Anarchy is equal to $2$.
\end{theorem}

\begin{proof}
   The proof follows from the same reasoning as the one for tribalism, except for a few computational differences. Since in all cases $i$ and $j$ are playing for the same ``tribe'', the increase utility received by $j$ also must be counted into the utility of the group.
   So in the case of edge $e$ where $\s_i^*(e) = B_i$ and $\s_j^*(e) = B_j$. We again mark $i$ and $j$ as witnesses for $e$, and have
   \[ U (\s) \geq U (\s_i^*; \s_j^*; \s) \geq U (\s) - 2 (u_i(\s) + u_j(\s)) + 2w_e (\s^*). \]
   For $e$ where $\s_i(e) = B_i$ and $\s_j (e) = 0$, we mark $j$  have
   \[ U (\s) \geq U (\s_i^*; \s) \geq U (\s) - 2 u_i (\s) + 2 w_e (\s^*). \]
   So $w_e(\s_*)$ is less than or equal to the utility of the witness of $e$ in strategy $\s$. Then we have
   \begin{align*}
    U(\s^*) = 2 \sum_e w_e(\s^*) \leq 2 \sum_{i \in V} u_i(\s) = 2 U(\s).
  \end{align*}

  The tightness of the upper bound can be shown with a simple example; consider a square graph with four nodes $\{a,b,c,d\}$ all with budget one, and four edges: $e_1 = \{a,b\}$, $e_2 = \{b,c\}$, $e_3 = \{c,d\}$, and $e_4 = \{d, a\}$. The reward functions are $f_{e_1}(x,y) = f_{e_3}(x,y) = (1- \varepsilon) xy$ and $f_{e_2}(x,y) = f_{e_4}(x,y) = \frac12 xy$. The arrangement that maximises social welfare is the one with each node investing their whole budget into the incident edge with payoff $(1-\epsilon)xy$. In this strategy, the social welfare is $4 (1-\epsilon)$. However, the strategy where every player invests in their incident edge with pay-off $\frac12 xy$ is an altruistic Nash equilibrium. Any unilateral deviation is clearly not beneficial, losing social welfare $2(1-\epsilon)$ for no gain. Bilateral deviations on $e_1$ and $e_3$ would lose society $2\epsilon$, and on $e_2$ and $e_4$ would lose society $2(1-\epsilon)$ utility. This strategy gives social welfare $2$, so the Price of Anarchy is at least $2$.
\end{proof}

\section{Tribes coordinating}
\label{sec:coord}

We briefly demonstrate how to interpret the approach of \cite{oliroute} with definitions designed to resemble our own.

\begin{definition} Let $G=(\P, (\Sigma_i)_{i \in \P}, (c_i)_{i\in \P})$ be a finite cost-minimisation game with a social cost function $C_G$.

Let $\tau:\P\rightarrow \mathds{N}$  
be a function that assigns each player a unique tribe, identified e.g. by a natural number.
Then the acting players are the oligopolies themselves, with cost function being sum of the cost of the players in the oligopoly.
Specifically, the \emph{$\tau$-oligopolistic extension of $G$}

is the cost-minimisation game $$oG^\tau = \left[\tau (\P), \left[ \prod_{\tau(i) = t} \Sigma_i \right]_{t \in \tau(\P)} ,(c_t^\tau)_{t \in \tau(\P)} \right],$$ where the cost experienced
by each oligopoly is the sum of costs of all players in the same oligopoly in the original game:
for every $t \in \tau(\P)$ and $\vec{s}\in \Sigma = ( \prod_{\tau(i) = t} \Sigma_i )_{t \in \tau(\P)} $,
$$ c_t^\tau(\vec{s}) = \sum_{i\in \P: t=\tau(j)} c_i(\vec{s}). $$
\label{defn:oliexten}
\end{definition}

\begin{definition}
  Given a $\tau$-oligopolistic extension of a finite cost minimisation game $$G=\left[\tau(\P), \left[\prod_{\tau(i) = t}\Sigma_i\right]_{t \in \tau(\P)}, (c_t^\tau)_{t \in \tau(\P)}\right],$$ the strategy $\vs \in \Sigma$ is a (pure) \emph{oligopolistic Nash equilibrium} if for all $i \in \P$ and $\vs' \in (\prod_{\tau(i) = t}\Sigma_i)_{t \in \tau(\P)}$,
  $$c_t^\tau (\vs) \leq c_t^\tau (\vs'_t;\vs_{-t}).$$
  \label{defn:olinash}
\end{definition}

\begin{definition}
The pure \emph{coordinated Price of Tribalism} (PoT), or oligopolistic Price of Anarchy, of a class of games $\mathcal{G}$ and class of partition functions for each game $\T=\{\T_G\}_{G\in \mathcal{G}}$ to be
$$ \PoT(\T,\mathcal{G}) = \sup_{G\in \mathcal{G}, \tau\in \T_G} \frac{ \sup_{\vec{s}\in S_{oG^\tau}} C_G(\vec{s}) }{ \inf_{\vec{s}\in \Sigma} C_G(\vec{s}) }, $$
where $S_{oG^\tau}$ is the set of pure Nash equilibria of $oG^\tau$.

\label{defn:opoa}
\end{definition}

\begin{theorem}
  The coordinated Price of Tribalism for the social $2$-grouping game is $2$.
  \label{thm:oPoAgrouping}
\end{theorem}

\begin{proof}
For the lower bound, consider the following configurations:
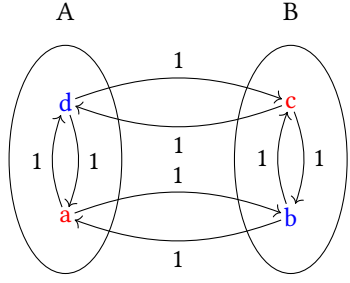
\begin{center}
\begin{tikzpicture}[scale=1.5]
 \node [inner sep=0.1em]  (a) at (0,0) {\color{red} a};
 \node [inner sep=0.1em]  (b) at (2,0) {\color{blue} b};
 \node [inner sep=0.1em]  (c) at (2,1) {\color{red} c};
 \node [inner sep=0.1em]  (d) at (0,1) {\color{blue} d};

 \node (A) at (0,1.8) {A};
 \draw (0,0.5) ellipse (0.5 and 1);

 \node (B) at (2,1.8) {B};
 \draw (2,0.5) ellipse (0.5 and 1);

 \draw[->] (a) to [bend left=20] node[above] {$1$} (b);
 \draw[->] (b) to [bend left=20] node[left] {$1$} (c);
 \draw[->] (c) to [bend left=20] node[below] {$1$} (d);
 \draw[->] (d) to [bend left=20] node[right] {$1$} (a);
 \draw[<-] (a) to [bend right=20] node[below] {$1$} (b);
 \draw[<-] (b) to [bend right=20] node[right] {$1$} (c);
 \draw[<-] (c) to [bend right=20] node[above] {$1$} (d);
 \draw[<-] (d) to [bend right=20] node[left] {$1$} (a);
\end{tikzpicture}
\end{center}
In this configuration each oligopoly has utility $4$, and if either oligopoly move their players into the same clique, they would continue to have utility $4$.
By symmetry, switching cliques of the both players also doesn't affect the utility.
So the oligopolistic Price of Anarchy is at least 2.

For the upper bound, first consider a fixed oligopoly $t$.
Let $\vec s = \Sigma_1 \times \dotsb \times \Sigma_N \to \{0,1\}$ be a Nash equilibrium assigning each member to a clique,
and $I_0 = \{ x \in \tau^{-1}(t) : \vec s (x) = 0\}$ the members of $t$ who are in clique $0$.
Then by the Nash condition, we know that if $t$ switches the clique of all members in $t$ of clique $0$ to clique $1$, it would only decrease the utility of $t$.
Thus, we have
\begin{align*}
  \sum_{i: \tau(i) = t} \sum_{j: \vec s (j) = \vec s (i)} u_{ji} &\geq
    \sum_{i: \tau(i) = t \land i \in I_0} \left( \sum_{j: \vec s (i) \neq \vec s (j)} u_{ji} + \sum_{j \in I_0} u_{ji} \right)  \\
    & \qquad + \sum_{i: \tau(i) = t \land i \not \in I_0 } \left( \sum_{j: \vec s (i) = \vec s (j)} u_{ji} + \sum_{j \in I_0} u_{ji} \right)
\end{align*}
where we assume that $u_{ii} = 0$. If we repeat the same argument to $\tau^{-1}(t) \setminus I_0$, and add it to the above inequality, we get
\begin{align*}
  2\sum_{i: \tau(i) = t} \sum_{j: \vec s (j) = \vec s (i)} u_{ji} &\geq
    \sum_{i: \tau(i) = t \land i \in I_0} \left( \sum_{j: \vec s (i) \neq \vec s (j)} u_{ji} + \sum_{j: \vec s (i) = \vec s (j)} u_{ji} + \sum_{j \in \tau^{-1}(t)} u_{ji} \right)  \\
    & \qquad + \sum_{i: \tau(i) = t \land i \not \in I_0 } \left( \sum_{j: \vec s (i) = \vec s (j)} u_{ji} + \sum_{j: \vec s (i) \neq \vec s (j)} u_{ji} + \sum_{j \in \tau^{-1}(t)} u_{ji} \right)
\end{align*}
Then for each oligopoly $t$, we have
\begin{equation}
2\sum_{i: \tau(i) = t} \sum_{j: \vec s (j) = \vec s (i)} u_{ji} \geq \sum_{i: \tau(i) = t} \sum_{j} u_{ji}
\end{equation}
Summing over $t$, we get $2 C(\vec s)$ on the left and $C(\vec s^*)$ on the right, so the oligopolistic Price of Anarchy is at most $2$.
\end{proof}

\end{document}